\documentclass [journal,onecolumn,11pt]{IEEEtran}
\usepackage{amsfonts,amsmath,amssymb}
\usepackage{indentfirst, setspace}
\usepackage{url,float}
\usepackage{color}
\usepackage{xcolor}
\usepackage[a4paper,ignoreall]{geometry}
\usepackage{longtable}
\usepackage{graphicx}
\usepackage[a4paper,ignoreall]{geometry}
\usepackage{multicol}
\usepackage{stfloats}
\usepackage{enumerate}
\usepackage{cite}
\usepackage{amsthm}
\usepackage{multirow}
\usepackage{amssymb}
\usepackage{galois}
\usepackage[all]{xy}
\usepackage[square, comma, sort&compress, numbers]{natbib}
\usepackage{url,doi}
\usepackage{natbib,hyperref,doi}
\hypersetup{hidelinks}
\usepackage{comment}
\usepackage{verbatim}
\usepackage{tikz}
\usetikzlibrary{arrows}
\usepackage[justification=centering]{caption}
\usepackage{amsmath}
\usepackage{cleveref}

\def\qu#1 {\fbox {\footnote {\ }}\ \footnotetext { From Qu: {\color{red}#1}}}
\def\kq#1 {\fbox {\footnote {\ }}\ \footnotetext { From Quan: {\color{blue}#1}}}
\def\tl#1 {\fbox {\footnote {\ }}\ \footnotetext { From Niu: {\color{blue}#1}}}
\def\wang#1 {\fbox {\footnote {\ }}\ \footnotetext { From Wang: {\color{purple}#1}}}

\newcommand{\mtl}[1]{{{\color{blue}#1}}}

\newtheorem{Th}{Theorem}[section]
\newtheorem{Cor}[Th]{Corollary}
\newtheorem{Prop}[Th]{Proposition}

\newtheorem{Lem}[Th]{Lemma}
\newtheorem{Def}[Th]{Definition}
\newtheorem{Exa}[Th]{Example}

\newcommand{\tr}{{\rm Tr}}

\newcommand{\gf}{{\mathbb F}}

\makeatletter
\newcommand{\figcaption}{\def\@captype{figure}\caption}
\newcommand{\tabcaption}{\def\@captype{table}\caption}
\makeatother

\begin{document}

	\title{More constructions of $n$-cycle permutations}

\author{
		{Tailin Niu, Kangquan Li, Longjiang Qu and Bing Sun}
		\thanks{
			This work is supported
			in part by the National Natural Science Foundation of China (NSFC) under Grant 62032009 and Grant 62172427,
			in part by the State Key Development Program for Basic Research of China under Grant 2019-JCJQ-ZD-351-00,
			in part by the Natural Science Foundation of Hunan Province of China under Grant 2021JJ40701,
			and in part by the Research Fund of National University of Defense Technology under Grant ZK22-14 and Grant ZK20-42.	\emph{(Corresponding author: Longjiang Qu.)}
			
			The authors are with the College of Science,
			National University of Defense Technology, Changsha, 410073, China (e-mail: runningniu@outlook.com; likangquan11@nudt.edu.cn; ljqu\_happy@hotmail.com; happy\_come@163.com).
			They are also with Hunan Engineering Research Center of Commercial Cryptography Theory and Technology Innovation, Changsha 410073, China.
			
				  }
		}

	\maketitle{}
	\begin{abstract}
		$n$-cycle permutations with small $n$ have the advantage that their compositional inverses are efficient in terms of implementation.
		They can be also used in constructing Bent functions and designing codes.
		Since the AGW Criterion was proposed, the permuting property of several forms of polynomials has been studied.
		In this paper, characterizations of several types of $n$-cycle permutations are investigated.
		Three criteria for $ n $-cycle permutations of the form $xh(\lambda(x))$, $ h(\psi(x)) \varphi(x)+g(\psi(x)) $ and  $g\left(  x^{q^i} -x +\delta   \right) +bx $ with general $n$ are provided.
		We demonstrate these criteria by providing explicit constructions.
		For the form of $x^rh(x^s)$, several new explicit triple-cycle permutations are also provided.
		Finally, we also consider  triple-cycle permutations of the form $x^t + c\tr_{q^m/q}(x^s)$ and provide one explicit construction.
		Many of our constructions are both new in the $n$-cycle property and the permutation property.
		
	\end{abstract}
	

%
	
	\section{Introduction}

		
		%
		%
		%

		%
		

	Let $ q $ be a prime power and $\gf_q$ be the finite field with $q$ elements.
	A polynomial $f(x) \in\gf_q[x]$ is called a \textit{permutation polynomial} (PP) and $f^{-1}$ denotes the compositional inverse of $ f $, if the map $f: a \mapsto f(a) $ is a bijection on $ \gf_q $.
	If there exists a positive integer $ n $ such that $ f^{(n)}=I $ is the identity map, $ f $ is called an \textit{n-cycle permutation}, where the $n$-th functional power of $ f $ is defined inductively by $ f^{(n)} = f \circ f^{(n-1)} = f^{(n-1)} \circ f $ and $ f^{(1)}=f, f^{(0)}=I ,f^{(-n)}=(f^{-1})^{(n)} $ with our notation.
	In this paper, $ n $-cycle permutations are called \textit{low-cycle permutations} for a small $ n $.
	When $n=2$ or $3$, $f$ is also called an \textit{involution}, or a \textit{triple-cycle permutation} respectively.

	Permutation polynomials over finite fields have wide applications in coding theory, cryptography, and combinatorial design theory, and we refer the readers to \cite{MullenWang14, hou2015permutation, li2018survey,Wang2019index,bartoli2021hasse,anbar2018carlitz} and the references therein for more details of the recent advances and contributions to the area.
	It is a challenging task to find new classes of permutation polynomials.
	However in 2011, Akbary et al. \cite{akbary2011constructing} provided a powerful method for constructing PPs over finite fields, which is called the AGW Criterion now.
	It both provided a unified explanation of earlier constructions of PPs and served a method to construct many new classes of PPs.
	After then, permutation polynomials of the form $x^{r}h(x^{s}) $ over $\gf_{q}$ were constructed by some researchers; see \cite{ding2015permutation,li2017new,gupta2016some,li2017several,zha2017further,li2018newp,wu2017permutation,xu2018some,tu2018two,li2017two,zhengTwoClassesPermutation2021a,liSeveralClassesComplete2021,liNewPermutationTrinomials2019b,houDeterminationClassPermutation2020,bartoliFamilyPermutationTrinomials2021a} etc. for more details. 
	Similarly, PPs of the form $xh(\lambda(x))$, $ h(\psi(x)) \varphi(x)+g(\psi(x)) $ and  $g\left(  x^{q^i} -x +\delta   \right) +bx $ were studied in \cite{akbary2011constructing,zheng2019two,zha2016some,tuxanidy2014inverses,yuan2007four} etc. 
	In addition, several authors researched constructions of the PPs $x+ c\tr_{q^{m} / q}(x^s)$; see e.g. \cite{kyureghyan2016permutation,charpin2010monomial,li2018permutation,gerikePermutationsFiniteFields2020,ZHA2019101573} for more details.
	In engineering, if both the permutation and its compositional inverse are efficient in terms of implementation, it is beneficial for the designer.
	This motivates the use of low-cycle permutations in the S-box of block ciphers.
	That the implementation of its inverse does not require much resources is a direct practical advantage of a low-cycle permutation.
	In devices with limited resources as a part of a block cipher, this is very useful.
	For instance, involutions have been used frequently in block cipher designs, in AES \cite{daemen2013design}, Khazad \cite{barreto2000khazad}, Anubis \cite{barreto2000anubis} and PRINCE \cite{borghoff2012prince}.
	Furthermore, low-cycle permutations (such as involutions) have been also used to construct Bent functions over finite fields \cite{coulter2018bent,gallager1962low} and to design codes \cite{gallager1962low}.
	In \cite{canteaut2015behaviors}, behaviors of permutations of an affine equivalent class have been analyzed with respect to some cryptanalytic attacks, and it is shown that low-cycle permutations (such as involutions)  are nice candidates against these attacks.
	In addition, in classifying permutations in the view of cycle, the research of $ n $-cycle permutations will be quite helpful, since each permutation over finite sets must be an $ n $-cycle permutation for at least one positive integer $n$.
	Because of the importance of $ n $-cycle permutations, in recent years, there are several studies about them.
	Charpin et al. \cite{charpin2016involutions} started the explicit study of involutions for finite fields with even characteristic. 
	Since then, a lot of attentions had been drawn in this direction.
	In 2019, a more concise criterion for involutory permutations of the form $x^rh(x^s)$ over $\gf_q$ was given by Zheng et al. \cite{zheng2019constructions}, where $s\mid {(q-1)}$.
	By using this criterion, from a cyclotomic perspective, they proposed a general method to construct involutions of such form from given involutions over some subgroups of $\mathbb{F}_{q}^{*}$ by solving congruent and linear equations over finite fields.
	Niu et al. \cite{niu2019new} started from the AGW Criterion, and proposed an involutory version of the AGW Criterion, independently.
	They also provided several explicit involutions of the forms $ x^rh(x^s) $ and $ g\left(x^{q^i} - x+\delta\right) +cx  $.	
	Monomial, Dickson polynomial and Linearized triple-cycle permutations over binary fields were studied by \cite{liuTripleCyclePermutationsFinite2019}.
	In 2020, Wu et al. \cite{wuCharacterizationsConstructionsTriplecycle2020a} generalized the work of \cite{zheng2019constructions} and obtained some characterizations of triple-cycle permutations of the form $ x^r h(x^s ) $.
	After that, Chen et al. \cite{chenConstructionsNcyclePermutations2021} generalized the work of \cite{wuCharacterizationsConstructionsTriplecycle2020a,zheng2019constructions,niu2019new,zheng2019constructions} and obtained criteria for $ n $-cycle permutations, which mainly are of the form $ x^rh(x^s) $.
	Chen et al. \cite{chenConstructionsNcyclePermutations2021} also proposed other constructing tools and several explicit triple-cycle permutations of the form $ x^rh(x^s) $ from both usual perspective and cyclotomic perspective.

	There are a lot of researches about the permutation property of several forms of polynomials $x^rh(x^s)$, $xh(\lambda(x))$, $ h(\psi(x)) \varphi(x)+g(\psi(x)) $, $g\left(  x^{q^i} -x +\delta   \right) +bx $ and  $x+ c\tr_{q^{m} / q}(x^s)$. 
	However, $n$-cycle permutation of the form $x^rh(x^s)$ have not been well investigated so far, and there are few studies about $n$-cycle property of other forms.
	Furthermore, explicit constructions of $n$-cycle permutations for general $n$ and $n=3$ is rarely found.
	New constructions that both new in $n$-cycle property and permutation property can also be obtained by researching $n$-cycle permutations.
	These motivate us to consider the characterizations of $n$-cycle property for several forms and to provide several new explicit constructions.
	The main purpose of this paper is to investigate general criteria for $ n $-cycle permutations of several forms over finite fields, and provides a way to acquire  cycle permutations from constructing non-identity mappings over subsets of finite fields.
	First, motivated by the AGW Criterion, we propose three criteria for $ n $-cycle permutations of the form $xh(\lambda(x))$, $ h(\psi(x)) \varphi(x)+g(\psi(x)) $ and  $g\left(  x^{q^i} -x +\delta   \right) +bx $ with general $n$.
	We also demonstrate these criteria by constructing explicit $ n $-cycle permutations with general $n$ and $n=3$.
	Then, for $x^rh(x^s)$ over $\gf_q$, we provide several explicit triple-cycle permutations by considering different $g(x)=x^rh(x)^s$ over $\mu_\ell=\left\{    x\in{\gf}_{q}^*  \    |  \   x^\ell=1   \right\}$, where $  \ell = {(q-1)/s} $.
	Finally, we consider triple-cycle permutations of the form $x^t + c\tr_{q^m/q}(x^s)$ and provide one construction.
	Many of explicit constructions in this paper are both new in $n$-cycle property and permutation property, especially those in Section \ref{fenyuan}.

	The rest of this paper is organized as follows.
	In Section \ref{pre},  we introduce some basic knowledge about $ n $-cycle permutations.
	Criteria for $ n $-cycle permutations of the form $xh(\lambda(x))$, $ h(\psi(x)) \varphi(x)+g(\psi(x)) $ and  $g\left(  x^{q^i} -x +\delta   \right) +bx $ are proposed in Section \ref{other}.
	Triple-cycle permutations of the form $x^rh(x^s)$ are constructed in Section \ref{fenyuan}.
	Furthermore, we provide an explicit triple-cycle permutations of the form $x^t + c\tr_{q^m/q}(x^s)$.

	\section{Preliminaries}\label{pre}
	
In this section, we introduce some basic knowledge.
	
	\begin{Def}
		If there exists a positive integer such that $ f^{(n)}=I $, we call $ f $ an \textit{$n$-cycle permutation}.
	\end{Def}

Monomial $n$-cycle permutations by Lemma \ref{monomial} will be basic components in obtaining some constructions.
	\begin{Lem}\label{monomial}
		Let $f(x)=x^{d}$ be a monomial polynomial over $\mathbb{F}_{q}$. 
		Then $f$ is $n$-cycle over $\gf_{q}$ if and only if $d^n \equiv 1 \pmod{q-1}$.
	\end{Lem}

	\begin{Lem}\label{xianxingkuochong}
		Assume $f(x) \in \gf_q[x]$ is an $n$-cycle permutation over $ \gf_{q^m}$.
		 If $m \mid ni$,  then $g(x)=f(x)^{q^i}$ is also an $n$-cycle permutation over $ \gf_{q^m}$.
	\end{Lem}
	\begin{proof}
		We have $f^{(n)}(x)=x$ and $f(x)^{q^i}=f(x^{q^i})$.
		Clearly, $g^{(n)}(x)=f^{(n)}(x^{q^{ni}} )=x^{q^{ni}}$.
		Thus, $g$ is also $n$-cycle.
	\end{proof}


	Let $m$ be a positive integer.
	We use $\operatorname{Tr}_{q^m/q}(\cdot)$ to denote the trace function form $\mathbb{F}_{q^{m}}$ to $\mathbb{F}_{q}$, i.e.,
	$$	\operatorname{Tr}_{q^m/q}(x)=x+x^{q}+x^{q^2}+\cdots+x^{q^{m-1}} .	$$
	We use $A^*$ to denote a set containing all nonzero elements of a set $ A $.
	The cardinality of a set $A$ is denoted by $ |  A  |$.
	For a mapping $f$, the kernel of $f$ is denoted by $\ker(f)$.

	Let  $F : \gf_{p^m} \rightarrow \gf_{p^m}  $  be a mapping, and $\omega$ be a $p$-th primitive unit root, where $p$ is a prime and $ m $ is a positive integer.
The Walsh transform of $F$ at $(u, v) \in \mathbb{F}_{p^{m}} \times \mathbb{F}_{p^{m}}$ equals by definition the Walsh transform of the so-called component function $\tr_{p^{n} / p}(v F(x))$ at $u$, that is:
$$
W_{F}(u, v):=\sum_{x \in \mathbb{F}_{p^{m}}}\omega^{\tr_{p^{m}/ p} (v F(x))+\tr_{p^{m}/ p}(u x) }.
$$
We have a proposition for involutions by the Walsh transform.
\begin{Prop}
	Assume $F $ permutes $ \gf_{p^m} $.
	Then, $F$ is an involution if an only if $ W_{F}(u, v)=W_{F}(v, u)	$ for each $(u, v) \in \mathbb{F}_{p^{m}} \times \mathbb{F}_{p^{m}}$.
\end{Prop}
\begin{proof}
	Assume $F$ is an involution. Then, we have 
	$$ 	W_{F}(u, v)=\sum_{x \in \mathbb{F}_{p^{m}}}\omega^{\tr_{p^{m}/ p} (v F(x))+\tr_{p^{m}/ p}(u x) } =\sum_{x \in \mathbb{F}_{p^{m}}}\omega^{\tr_{p^{m}/ p} (v F(F(x)))+\tr_{p^{m}/ p}(u F(x)) } =W_{F}(v, u) . 	$$
	Conversely, we assume $ W_{F}(u, v)=W_{F}(v, u)	$.
	Since  $$ 	W_{F^{-1}}(v,u)=\sum_{x \in \mathbb{F}_{p^{m}}}\omega^{\tr_{p^{m}/ p} (u F^{-1}(x))+\tr_{p^{m}/ p}(v x) } =\sum_{x \in \mathbb{F}_{p^{m}}}\omega^{\tr_{p^{m}/ p} (u x)+\tr_{p^{m}/ p}(v F(x)) } =W_{F}(u, v)   ,	$$
	we arrive at $  W_{F^{-1}}(v,u)=  W_{F}(v,u) $, for each $(u, v) \in \mathbb{F}_{p^{m}} \times \mathbb{F}_{p^{m}}$.
	Thus, $F$ is an involution.
\end{proof}
	
	When we construct cycle permutations, the following result is inspiring and useful.
	\begin{Lem}
		\label{LGWlemma}
		(\cite{akbary2011constructing}, AGW Criterion)
		Let $A, S$, and $\overline{S}$ be finite sets with $\# S=\# \overline{S}$, and let $f: A\to A,$ $g: S\to \overline{S}$, $\lambda: A\to S$ and $\overline{\lambda}: A\to\overline{S}$ be maps such that $\bar{\lambda}\circ f=g\circ \lambda$. 
		If both $\lambda$ and $\bar{\lambda}$ are surjective, then the following statements are equivalent:
		\begin{enumerate}[(1)]
			\item $f$ is a bijection and
			\item $g$ is a bijection from $S$ to $\overline{S}$ and $f$ is injective on $\lambda^{-1}(s)$ for each $s\in S$.
		\end{enumerate}
	\end{Lem}
	The AGW Criterion can be illustrated as the following commutative diagram:
	\begin{equation*}
		\xymatrix{
			A \ar[rr]^{f}\ar[d]_{\lambda} &   &  A  \ar[d]^{\overline{\lambda}} \\
			S	 \ar[rr]^{g} &  & \overline{S} }
	\end{equation*}
	Since the AGW Criterion was put forward, a lot of families of PPs were constructed by it.
In this paper,	a permutation polynomial is called an \textit{AGW-PP} if it is based on the AGW Criterion.


	%
	%

	\section{$n$-cycle permutations of three types of AGW-PPs}
	\label{other}

	In this section, we present three constructions of $n$-cycle permutations of the form $xh(\lambda(x))$, $ h(\psi(x)) \varphi(x)+g(\psi(x)) $ and  $g\left(  x^{q^i} -x +\delta   \right) +bx $.
	Most of them are  new in $n$-cycle property.
	
	\subsection{$n$-cycle permutations of the form $xh(\lambda(x))$}
	In \cite[Theorem 6.3]{akbary2011constructing}, Akbary et al. studied the permutation property of $xh(\lambda(x))$. 
	In this subsection, we consider the $n$-cycle property of permutations with the form $xh(\lambda(x))$, and several constructions are provided.
	
	\begin{Th}
		\label{snsxjdchsn}
		Let $ q $ be any power of a prime number $ p $, $m $ be any positive integer, and $ S $ be any subset of $ \gf_{q^m} $ containing $0$.
		Let $ h,k \in \gf_{q^m}$ be any polynomials such that $ h(0) \ne 0 $, $ k(0) = 0 $ and $ g(x) = xk(h(x)) $ permuting  $ \lambda(\gf_{q^m}) $.
		Let $ \lambda(x) \in \gf_{q^m}[x]  $ be any polynomial satisfying
		\begin{enumerate}[(1)]
			\item $h(\lambda(\gf_{q^m})) \subseteq S$; and
			\item $ \lambda(a\alpha) = k(a) \lambda(\alpha) $ for all $ a \in S $ and all $ \alpha \in \gf_{q^m}$.
		\end{enumerate}
		Then the polynomial $ f (x) = xh(\lambda(x)) $ is an $n$-cycle permutation if and only if
		\begin{equation}
			\label{asdasfasdf}
			\prod_{i=0}^{n-1} 	h\left( g^{(i)}(  y  )    \right)=1
		\end{equation}	
		holds for any $y \ne 0 \in  \lambda(\gf_{q^m}) $.
		Furthermore,
		\begin{equation}
			\label{kjbkybiub}
			\prod_{i=0   }^{n-1} k\left(   h\left(   g^{(i)}(  y  )  \right)  \right)=1
		\end{equation}	
		is necessary for $ f  $ being an $n$-cycle permutation.
	\end{Th}
	\begin{proof}
		Let $ 0 < s < n$.
		For  any $x\in \gf_{q^m}$ and $ y \in  \lambda(\gf_{q^m}) $ satisfying $y=\lambda(x)$, we have
		\begin{equation}\label{sdasdasefegv}
			\begin{aligned}			
				f^{(n)}(x)=&       f^{(n-1)}(   f(x)       )       \\
				=&      f^{(n-2)}\left( f(x)    h(\lambda(   f(x)      ) )  \right)   \\
				=&        f^{(n-2)}\left(xh(\lambda(x)) h(\lambda(xh(\lambda(x)) ))  \right) ,  \\ 
			\end{aligned}
		\end{equation}			
		by plugging $f(x)= xh(\lambda(x))$ into $f^{(n)}(x)$.
		After that, apply $ \lambda(   h(\lambda(x))     x) = k\left(h(\lambda(x)) \right) \lambda(x) $ into Eq. (\ref{sdasdasefegv}) and we obtain
		\begin{equation*}
			\begin{aligned}	
				f^{(n)}(x)=	 &     f^{(n-2)}\left( xh(\lambda(x))   h(       k(h(\lambda(x)))        \lambda(x )) \right)    \\
				=&      f^{(n-2)}\left( xh(   \lambda(x)    )   h(     g(   \lambda(x )  )    ) \right)    \\
				=&   	f^{(n-2)}\left( x    \prod_{i=0}^{1} 	h\left( g^{(i)}(  \lambda(x)  )    \right)     \right).
			\end{aligned}
		\end{equation*}
		So on and so forth, we arrive at the following
		\begin{equation}\label{sdahbnkjn}
			f^{(n)}(x) =f^{(n-s+1)}\left( x    \prod_{i=0}^{s-2} 	h\left( g^{(i)}(  \lambda(x)  )    \right)     \right),
		\end{equation}
		where $s$ is a positive integer.
		After plugging $f(x)= xh(\lambda(x))$ into Eq. (\ref{sdahbnkjn}),  we acquire
				\begin{equation} \label{kujbgjyqvud}
			\begin{aligned}			
					f^{(n)}(x)=&          f^{(n-s)}\left(  xh(\lambda(x))   \prod_{i=0}^{s-2} 	h\left( g^{(i)}(  \lambda(xh(\lambda(x)))  )    \right)   \right) . \\
								\end{aligned}
				\end{equation}
		Plugging $ \lambda(xh(\lambda(x))) = k(h(\lambda(x))) \lambda(x) $ and $ g(x) = xk(h(x)) $ into Eq. (\ref{kujbgjyqvud}), one can get
		\begin{equation*}
			\begin{aligned}			
				f^{(n)}(x)=&          f^{(n-s)}\left(  xh(\lambda(x))   \prod_{i=0}^{s-2} 	h\left( g^{(i)}(  \lambda(x)  k(h(\lambda(x)))     )    \right)   \right)\\
				=&          f^{(n-s)}\left(  xh(\lambda(x))   \prod_{i=0}^{s-2} 	h\left( g^{(i+1)}(    \lambda(x)  )    \right)   \right)\\
				=&          f^{(n-s)}\left(         x    \prod_{i=0}^{s-1} 	h\left( g^{(i)}(  \lambda(x)  )    \right)           \right).\\			
			\end{aligned}
		\end{equation*}
		Similarly, we finally arrive at
		\begin{equation}\label{kihbjghqbikjndq}
			\begin{aligned}
				f^{(n)}(x)			=&           f \left(         x    \prod_{i=0}^{n-2} 	h\left( g^{(i)}(  \lambda(x)  )    \right)           \right) =               x    \prod_{i=0}^{n-1} 	h\left( g^{(i)}(  \lambda(x)  )    \right)     .    \\		
			\end{aligned}
		\end{equation}
		
		On the one hand, assume for any $y\in  \lambda(\gf_{q^m}) $, $\prod_{i=0}^{n-1} 	h\left( g^{(i)}(  y  )    \right)=1 .$
		Then, according to Eq. (\ref{kihbjghqbikjndq}), $ f $ is an $n$-cycle permutation over $\gf_{q^m}$.
		On the other hand, assume that  $ f $ is an $n$-cycle permutation over $ \gf_{q^m}$.
		For each $y \in  \lambda(\gf_{q^m}) ^*  $, there exists an $x_0 \in \gf_q^*$ such that $\lambda(x_0)=y$.
		According to Eq. (\ref{kihbjghqbikjndq}), we have $\prod_{i=0}^{n-1} 	h\left( g^{(i)}(  y  )    \right)=1 .$
		Thus $f$ is an $n$-cycle permutation if and only if Eq. (\ref{asdasfasdf}) holds.
		
		Furthermore, we assume $f$ is an $n$-cycle permutation.
		For each $y \in  \lambda(\gf_{q^m})  $, there exists an $x_0 \in \gf_q$ such that $\lambda(x_0)=y$.
		For any $x\in \gf_{q^m}$, we have the following equation according to Eq. (\ref{kihbjghqbikjndq}):
		$$   \lambda\left(	  x_0    \prod_{i=0}^{n-1} 	h\left( g^{(i)}(  y )    \right)         \right)=\lambda(x_0) .$$
		Since $\prod_{i=0}^{n-1} 	h\left( g^{(i)}(  y  )    \right)  \in S$, one can obtain
				\begin{equation*}
			\begin{aligned}
			\prod_{i=0}^{n-1} k\left( 	h\left( g^{(i)}(  y  )  \right)   \right)    \lambda(  x_0  )= &\lambda(x_0)   . \\
					\end{aligned}
	\end{equation*}
		Thus, Eq. (\ref{kjbkybiub}) is necessary for $ f  $ being an $n$-cycle permutation.
	\end{proof}

	\begin{Prop}
		\label{hjksdhgjshs}
		Assume $ q $ is a prime power, $ m,n $ are positive integers, and $ h(x) \in \gf_{q}[x]$ is a polynomial such that for any $y \in \gf_q$, $ h(y)^n =1 $.
		Let $ \lambda(x) \in \gf_{q^m}[x]  $ be either $\lambda_1(x)=\tr_{q^m/q}(x^n) $ or $\lambda_2(x)=\sum\limits_{0\le i_1<i_2<\cdots < i_n\le m - 1} {{x^{{q^{i_1}}+{q^{i_2}}+ \cdots + {q^{i_n}}}}}$.
		Then the polynomial $ f (x) = xh(\lambda(x)) $ is an $n$-cycle permutation over $\gf_{q^m}$.
	\end{Prop}
	\begin{proof}
		In Theorem \ref{snsxjdchsn}, we have $\lambda(\gf_{q^m}) = \gf_q$ and $\lambda(a\alpha)=a^n\lambda(\alpha)$ for all $a \in \gf_{q}$ and $\alpha \in \gf_{q^m}$.
		Since $h(y)^n=1$ holds for $y \in \gf_q$, $ g(y)=yh(y)^n=y $ is an $n$-cycle permutation over $\gf_q$.
		Plugging $h(y)^n=1$ and $ g(y)=y $ into Eq. (\ref{asdasfasdf}), one can get that $f$ is an  $n$-cycle permutation, according to Theorem \ref{snsxjdchsn}.
	\end{proof}

	There are a lot of polynomials $h$ satisfying $h(y)^n=1$, for any $y \in \gf_q$.
	Below are some examples.

	\begin{Cor}
		Let  $q$ be a prime power, $n$ be a positive integer satisfying $n \mid (q-1)$.
		Then, the polynomial $$ f (x) = x\left(     1 +   \theta  \lambda(x)^{(q-1)/n}-  \lambda(x)^{q-1}        \right) $$ is an $n$-cycle permutation over $\gf_{q^m}$, where $ \lambda$ is either $\lambda_1$ or $\lambda_2$ in Proposition \ref{hjksdhgjshs}, $\theta$ is an $n$-th primitive unit root and $m $ is a positive integer.
	\end{Cor}
	\begin{proof}
		Let $h(y)=1 + \theta y^{(q-1)/n}- y^{q-1}  $ and $f(x)$ can be written as $xh(\lambda(x))$.
		In the following, we will prove $h(y)^n =1$, for $y \in \gf_{q} $.
		First, $h(0)=1$.
		Next, for $y \ne 0$, we have
		$		h(y)^n=    \left( 1 +  \theta y^{(q-1)/n}-1   \right)^n    =      y^{q-1}   =      1 .  $
		Thus, $f$ is an $n$-cycle permutation according to Proposition \ref{hjksdhgjshs}.
	\end{proof}

	When $n=2$, there are also some involutions that easy to be obtained.
	\begin{Cor}\label{njavhwds}
		Let $q$ be an odd prime power.
		The polynomial $$ f (x) = x\left(     1-2  \lambda(x)^{q-1}         \right) $$ is an involution on $\gf_{q^m}$, where $ \lambda$ is either $\lambda_1$ or $\lambda_2$ in Proposition \ref{hjksdhgjshs} and $m $ is a positive integer.
	\end{Cor}
	\begin{proof}
		Let $h(y)=   1 -2  y^{q-1} $.
		Then $f(x)$ can be written as $xh(\lambda(x))$.
		We have $h(0)=1$ and $h(y)=-1$, for $y \ne 0$.
		Thus, $f$ is an involution according to Proposition \ref{hjksdhgjshs}.
	\end{proof}
	Corollary \ref{njavhwds} is a generalization of Example 5.5 in \cite{2022niufinding}.

	\begin{Cor}\label{kjbsdhvbaihdviqwb}
		Let $q$ be a power of odd prime $p$.
		Assume $a,b,c$ are integers satisfying $a^2 +b^2 \equiv 0 \pmod p  $ and $4c \equiv 0 \pmod {q-1} $.
		Then, the polynomial $$ f (x) = x      \left(  1+   a \lambda(x)^c  +b\lambda(x)^{q-c-1} -\lambda(x)^{q-1}     \right)       $$ is an involution on $\gf_{q^m}$, where $ \lambda$ is either $\lambda_1$ or $\lambda_2$ in Proposition \ref{hjksdhgjshs} and $ m $ is a positive integer.
	\end{Cor}
	\begin{proof}
		Let $h(y)= 1+   a y^c  +by^{q-c-1} -y^{q-1}   $ and $f(x)$ can be written as $xh(\lambda(x))$.
		For $y \ne 0$, we have	
		$
		h(y)^2=    	 \left(    a y^c  +by^{q-c-1}  \right)^2
		=    	  a^2y^{2c}  +b^2y^{-2c}      +2ab
		=      1  ,
		$
		where $a^2y^{2c}  +b^2y^{-2c} =0$ according to $4c \equiv 0 \pmod {q-1} $ and $a^2  \equiv -b^2 \pmod p$.
		Thus $h(y)^2 =1$ holds for $y \in \gf_{q} $.
		Thus, $f$ is an involution according to Proposition \ref{hjksdhgjshs}.
	\end{proof}

	Here, we provide simple examples for Corollary \ref{kjbsdhvbaihdviqwb}.
	Let $h(y)=1 + y+ 3 y^3+ 4 y^4$.
	Then $f(x)$ can be written as $xh(\lambda(x))$.
	We have $h(y)=1$, for $y=0,2,4$ and $h(y)=-1$ for $y=1,3,5$.
	Thus, the polynomial $ f (x) = x\left(    1 + \lambda(x)+ 3 \lambda(x)^3+ 4 \lambda(x)y^4          \right) $ is an involution on $\gf_{5^m}$ according to Proposition \ref{hjksdhgjshs}, where $ \lambda(x) =\sum\limits_{0\le i < j\le m - 1} {{x^{{5^i} + {5^j}}}}  $ and $m $ is a positive integer.

	\subsection{$n$-cycle permutations of the form $ h(\psi(x)) \varphi(x)+g(\psi(x)) $}
	In \cite[Theorem 5.1]{akbary2011constructing}, Akbary et al. investigated the permutation property of $ h(\psi(x)) \varphi(x)+g(\psi(x)) $.
	In this subsection, we consider the $n$-cycle property of permutations with the form of $ h(\psi(x)) \varphi(x)+g(\psi(x)) $.
	Several constructions are given.

		\begin{Th}
			\label{kjhbjhvgujygvfuyv}
			Consider any polynomial $g \in \mathbb{F}_{q^m}[x]$, any $q$-polynomials $\varphi, \psi \in \mathbb{F}_{q^m}[x]$ satisfying that $ \varphi$ is an $n$-cycle permutation over $\gf_{q^m}$ and $\varphi \circ \psi={\psi} \circ \varphi$.
			Then $$ f(x)= \varphi(x)+g(\psi(x)) $$ is an $n$-cycle permutation over $\mathbb{F}_{q^m}$ if and only if
			\begin{equation}\label{asuyhdgiaysdjgavsdas}
				\sum_{k=0}^{n-1} \varphi^{(\mtl{n-1-k})} \left(    g(   \bar{f}^{(k)}(  y )         )       \right) =0
			\end{equation}
			holds for any $y \in \psi(\gf_{q^m}) $, where $\bar{f}(x)= \varphi(x)+{\psi}(g(x))$.
		\end{Th}
		\begin{proof}	
			For  any $x \in \gf_{q^m}$,  and $y \in \psi(\gf_{q^m}) $ such that $y= \psi(x)$,
			we have
%
%
%
%
%
			\begin{equation*}
				\begin{aligned}
					f^{(n)}(x)		=&           \varphi(         f\circ f^{(n-2)}(x)           )+g(\psi \circ    f\circ f^{(n-2) }  \mtl{(x)}       )       \\
					=&           \varphi(      \varphi(f^{(n-2)}(x))  )      +         \varphi(   g(  \psi(f^{(n-2)}(x))   )           )+g(  \bar{f} \circ \psi  \circ f^{(n-2) }(x)      )       \\	
					=& \sum_{k=0}^{1} \varphi^{(\mtl{1-k})} \left(    g(   \bar{f}^{(k)}(  \psi(f^{(n-2)}(x))  )         )       \right) +   \varphi^{(2)}(f^{(n-2)}(x)). \\
				\end{aligned}
			  \end{equation*}
		 So on and so forth, this will lead to
		  \begin{equation*}\label{}
		  	\begin{aligned}
		  		f^{(n)}(x)			=&    \sum_{k=0}^{2} \varphi^{(\mtl{2-k})} \left(    g(   \bar{f}^{(k)}(  \psi(f^{(n-3)}(x))  )         )       \right) +   \varphi^{(3)}(f^{(n-3)}(x)) , \\
		  	\end{aligned}
		  \end{equation*}
	  	and finally arrive at
	  \begin{equation}\label{fgegfvcssfef}
	  	\begin{aligned}
	  		f^{(n)}(x)		=& \sum_{k=0}^{n-1} \varphi^{(\mtl{n-1-k})} \left(    g(   \bar{f}^{(k)}(  \psi(x)  )         )       \right) +   \varphi^{(n)}(x)  , \\
=& \sum_{k=0}^{n-1} \varphi^{(\mtl{n-1-k})} \left(    g(   \bar{f}^{(k)}(  y )         )       \right) + x , \\
	  	\end{aligned}
	  \end{equation}
which indicates that  $ f $ is an $n$-cycle permutation over $ \gf_{q^m} $.

		On the one hand, assume that Eq. (\ref{asuyhdgiaysdjgavsdas}) holds for any $y \in \psi(\gf_{q^m}) $.
Then, according to Eq. (\ref{fgegfvcssfef}), $ f $ is an $n$-cycle permutation over $\gf_{q^m}$.
On the other hand, assume that  $ f $ is an $n$-cycle permutation over $ \gf_{q^m}$.
For each $y \in \psi(\gf_{q^m}) $, there exists an $x_0 \in \gf_q$ such that $\psi(x_0)=y$.
According to Eq. (\ref{fgegfvcssfef}), we have $	\sum_{k=0}^{n-1} \varphi^{(\mtl{n-1-k})} \left(    g(   \bar{f}^{(k)}(  y )         )       \right) =  0 .$
Thus $f$ is an $n$-cycle permutation if and only if Eq. (\ref{asuyhdgiaysdjgavsdas}) holds  for any $y \in \psi(\gf_{q^m}) $.
\end{proof}		
		The proposition below is a generalization of involutory criterion in \cite{2022niufinding}.
		\begin{Prop}
			\label{akbary2011constructingwithlinearBi}
			Define $\psi,  g$  as in Theorem \ref{kjhbjhvgujygvfuyv}.
			Consider any $q$-polynomial $ \psi \in \mathbb{F}_{q^m}[x]$ satisfying $g(   \psi(\gf_{q^m})  ) \ne \{0\}$.
			Assume  $g(x) \in \mathbb{F}_{q^m}[x]$ is nonzero such that $g(\gf_{q^m}) \subseteq \ker(\psi)$.
			Then,
			$$f(x)=  x +  g(\psi(x) )   $$
			is an $n$-cycle permutation over $ \gf_{q^m} $ if and only if $p$ is  of of $n$, where $ p $ is the characteristic of $\gf_{q^m} $.
		\end{Prop}
		\begin{proof}
			Clearly $\ker(\varphi) \cap \ker (\psi)=\{0\} $, where $\varphi(x)=x$.
			Together with $g(\gf_{q^m}) \subseteq \ker(\psi)$, we have $\bar{f}(x)= x+\psi(g(x))=x$.
			According to Theorem \ref{kjhbjhvgujygvfuyv}, $f$ is an $n$-cycle permutation if and only if $ n  g(   y    )   =0$, which is equivalent to the condition that $p$ is a factor of $n$.
		\end{proof}
		
		Generally speaking, it is not hard to obtain $\varphi$ and $\psi$ satisfying $\varphi \circ \psi={\psi} \circ \varphi$.
		For example, both $\psi$ and $\varphi$ are $q$-polynomials over $\mathbb{F}_{q}$.
		In the corollary below, 	note that $\varphi(x)$ permutes $\gf_q$ due to $\sum_{i=0}^{m-1} \alpha_i \ne 0 $.

		\begin{Cor}
			\label{akbary2011constructingwithlinearBix}
			Assume $q$-polynomial $\psi$ satisfying  $\psi(\gf_{q})=\{0\}$.
			Let $H(x) $ be a nonzero polynomial over $\gf_{q^m}$ and $g(x)$ be either $g_1(x)=\tr_{q^m/q}(H(x))$ with $H(\psi(\gf_{q^m})) \not\subset \ker(\tr_{q^m/q}) $ or $g_2(x)=H(x)^s$ with $H(\psi(\gf_{q^m})) \ne \{0\}$, where positive integer $ s $ satisfies $ s(q-1) \equiv 0 \pmod{q^m-1} $.
			Then,
			$$f(x)=  x+  g(\psi(x) )   $$
			is an $n$-cycle permutation over $ \gf_{q^m} $ if and only if $p$ is a factor of $n$, where $ p $ is the characteristic of $\gf_{q^m} $.
		\end{Cor}
		\begin{proof}
			One can obtain $g(\gf_{q^m}) \subseteq \ker(\psi)$ by the expression of $g$.
			Thus $f$ is an $n$-cycle permutation if and only if $p$ is a factor of $n$.
		\end{proof}

		
		Such conditions in Corollary \ref{akbary2011constructingwithlinearBix} are not hard to meet.
				\begin{Exa}
Let $q$ be a power of $3$.
Then 	$x+ \tr_{q^m/q}\left(   (      x^q-x  )^2  \right)     $
is a triple-cycle permutation over $\gf_{q^m}$. 
		\end{Exa}

		\begin{Cor}
			Assume $s$ is an integer and $ c\in \gf_{q^2}^*$ satisfying $c+c^q=0$.
			Then, $ f(x)=  x+c \tr_{q^2/q}(x)^s   $ is a triple-cycle permutation over $\gf_{q^2}$ if and only if $q$ is a power of $3$.
		\end{Cor}
		\begin{proof}	
			We have $$\bar{f}(y)=y+\tr_{q^2/q}( cx^s)=y+  cx^s  + c^qx^{s}  =y  .$$ 
			Then, $				\sum_{k=0}^{2-1} g\left(       \bar{f}^{(k)}(  y )         \right) =   3 cy^s.$
			According to Proposition \ref{akbary2011constructingwithlinearBi}, the result is established.
		\end{proof}


		\newpage
		
		\begin{Prop}
			Consider any polynomial $g(x)=\sum_{t=0}^{q^3-2}  a_t x^t \in \mathbb{F}_{q^3}[x]$, and for each $t$, $\tr_{q^3/q}(a_t) =0$.
			Then $$ f(x)= x^q + g(\tr_{q^3/q}(x)) $$ is a triple-cycle permutation over $\mathbb{F}_{q^3}$.
		\end{Prop}
		\begin{proof}
			In Theorem \ref{kjhbjhvgujygvfuyv}, let $\varphi(x)=x^q$ and $\psi(x)=\tr_{q^3/q}(x)$.
			Clearly we have $\varphi \circ \psi={\psi} \circ \varphi$.
			Then, one can verify that $\tr_{q^3/q}(g(x))=0$ holds for any $x \in \gf_{q}$, since $\tr_{q^3/q}(a_t) =0$.
			Thus, $ \bar{f}(x)= x^q +\tr_{q^3/q}(g(x))=x^q$.
			Then, for any $y \in \gf_q$, 
				\begin{equation*}\label{}
					\mtl{\left( 	g(    y       ) \right)^{q^2}                    +	\left(    g(    y^q         )       \right) ^q   +   g(   y^{q^2}         )        =\tr_{q^3/q}(g(y))} =0,
				\end{equation*}
				which is equivalent to
				\begin{equation*}\label{}
					\sum_{k=0}^{2} \varphi^{(\mtl{2-k})} \left(    g(   \bar{f}^{(k)}(  y )         )       \right) =0.
				\end{equation*}
				According to Theorem \ref{kjhbjhvgujygvfuyv}, $f$ is a triple-cycle permutation over $\mathbb{F}_{q^3}$.
			\end{proof}

			%
			%
			%
			%
			%
			%
			%
			
			\subsection{$n$-cycle permutations of the form $g\left(  x^{q^i} -x +\delta   \right) +bx $}
				In \cite[Proposition 3]{zheng2019two}, Zheng et al. investigated the permutation property between $ g\left(  x^{q^i} -x +\delta   \right) +bx $ and $ g(x)^{q^i}-g(x)+bx $. 
				Niu et al. \cite{niu2019new} got an involutory version using compositional inverses.	
				In this subsection, we consider the $n$-cycle property of permutations with the form $g\left(  x^{q^i} -x +\delta   \right) +bx $.
				Some constructions are also provided.

			For their $ n $-cycle properties, we have the following results, similarly with above subsections.
			\begin{Th}
				\label{jingdianiff}
				Let $ \gf_{q^m} $ be the degree $ m $ extension of the finite field $ \gf_q $ and $ \delta \in \gf_{q^m} $, $g(x)\in\gf_{q^m}[x]$.
				Then $ f(x) =g( x^{q^i} -x +\delta)+x $ is an $n$-cycle permutation over $ \gf_{q^m} $ if and only if
				\begin{equation}\label{jhuvytcuhbkijb}
					\sum_{k=0}^{n-1} g\left(  h^{(k)}(y)     \right)   =0
				\end{equation}
				holds for any $y \in S_{\delta} =\left\{   x^{q^i} -x +\delta    \  | \  x \in \gf_{q^m}\right\} $, where $ h(y) = g(y)^{q^i} - g(y) + y  $ is on $S_{\delta}  $ , $ i $ is an integer with $ 1 \le i \le m-1$ and $\ell=\gcd(i,m)$.
			\end{Th}
			\begin{proof}
				Its proof is similar with that in Theorem \ref{kjhbjhvgujygvfuyv}, and thus it is omitted.
			\end{proof}

			\begin{Prop}
				\label{Zconstruction}
				Let  $ m, i $ be integers with $ 1 \le i \le m-1 $, $\ell=\gcd(i,m)$ and $ \gf_{q^m} $ be the finite field containing $q^m$ elements.
				Assume $ \delta \in \gf_{q^m} $ and nonzero polynomial $g(x)\in\gf_{q^m}[x]$ satisfying $g( \gf_{q^m}      ) \subseteq  \gf_{q^i}$ and $g(S_{\delta}) \ne \{0\}$, where $ S_{\delta} =\left\{   x^{q^i} -x +\delta   \  | \  x \in \gf_{q^m}\right\}$.
				Then, $$f(x)=   x +    g(x^{q^i} -x +\delta  )       $$
				is an $n$-cycle permutation over $ \gf_{q^m} $ if and only if $p$ is a factor of $n$, where $ p $ is the characteristic of $\gf_{q^m} $.
			\end{Prop}
			\begin{proof}
				Its proof follows in a similar manner with that in Proposition \ref{akbary2011constructingwithlinearBi}, and thus it is omitted.
			\end{proof}
				%

			\begin{Cor}
				\label{hsdcuasydgfcus}
				Assume integers $ m,i $ satisfy $ 1 \le i \le m-1 $.
				Let $H(x) $ be a nonzero polynomial over $\gf_{q^m}$ and $g(x)$ be either $g_1(x)=\tr_{q^m/q^i}(H(x))$ (with $H(S_{\delta}) \not\subset \ker(\tr_{q^m/q^i}) $) or $g_2(x)=H(x)^s$ (with $H(S_{\delta}) \ne \{0\}$), where positive integer $ s $ satisfies $ s(q^i-1) \equiv 0 \pmod{q^m-1} $.
				Then, for any $ \delta \in \gf_{q^m} $, $$f(x)=   x +   g(x^{q^i} -x +\delta  )       $$
				is an $n$-cycle permutation if and only if $p$ is a factor of $n$, where $ p $ is the characteristic of $\gf_{q^m} $.
			\end{Cor}

			\begin{Exa}
				Let $q$ be a power of $3$, integers $ m,i $ satisfy $ 1 \le i \le m-1 $, and $s= 1+q^i+q^{2i}+\cdots+  q^{m-i} $.
				For any $ \delta \in \gf_{q^m} $,	$  x +   (x^{q^i} -x +\delta  )^s       $
				is a triple-cycle permutation of $\gf_{q^m}$, where $c \in \gf_{q^\ell}^* \setminus \{1\}$.
			\end{Exa}


			\section{Triple-cycle permutations of the form $x^rh(x^s)$}
			\label{fenyuan}
			
			In this section, we provide triple-cycle permutations of the form $x^rh(x^s)$.
	       First, we  recall a lemma and simply derive another one. 
			
			\begin{Lem}
				\label{mulcore}
				\cite[Theorem 1]{wuCharacterizationsConstructionsTriplecycle2020a}
				Let $q$ be a prime power and $f(x)= x^rh\left( x^s\right)  \in \gf_q[x]$, where $s \mid (q-1),\gcd(r,s)=1$.
				Assume that $g(x)=x^rh(x)^s$ is a polynomial on $\mu_\ell=\left\{    x\in{\gf}_{q}^*  \    |  \   x^\ell=1   \right\}$, where $  \ell = {(q-1)/s} $.
				Then, $f$ is a triple-cycle permutation over $\mathbb{F}_{q}$ if and only if
				\begin{enumerate}[(1)]
					\item $r^{3} \equiv 1 \bmod s$ and
					\item $\varphi(y)=y^{(r^3-1)/s}h(y)^{r^2}h\left( g\left( y\right) \right)  ^rh\left( g\left( g(y) \right) \right)  =1 $ for all $y \in \mu_{\ell}$.
				\end{enumerate}
			\end{Lem}

			The lemma below is not hard to obtain.
			\begin{Lem}
				\label{single}
				Let $q$ be a prime power, $s \mid (q-1)$, $\gcd(r,s)=1$ and $r^{3} \equiv 1 \bmod s$.
				Assume that $h(x)\in \gf_q[x]$ such that $h(y)^s=ay^{v-r}$ holds for any $y \in \mu_\ell=\left\{    x\in{\gf}_{q}^*  \    |  \   x^\ell=1   \right\}$ , where $v^{3} \equiv 1 \bmod \ell$,  $a^{v^2+v+1}=1$ and $  \ell = {(q-1)/s} $.
				Then $f(x)= x^rh(x^s) $ is a triple-cycle permutation over $\mathbb{F}_{q}$  if and only if for any $y \in \mu_\ell$,
				$$ y^{(r^3-1)/s}h(y)^{r^2}h( ay^v )  ^rh\left(   a^{v+1}  y^{v^2}  \right)  =1 .$$
			\end{Lem}

			Note that  if $f(x) \in \gf_q[x]$ is a triple-cycle permutation over $ \gf_{q^3}$, then so does $f(x)^{q^i}$ for $i\in \{0,1,2\}$, according to Lemma \ref{xianxingkuochong}.
			\begin{Prop}
				Let $ q =2^{3m}$, where $m$ is a positive integer.
				Assume that integer $k$ satisfies $ 7k \equiv 0 \pmod{q-1}   $ and $k \equiv 3 \pmod 7$.
				Then,
				$$f(x)=x \left(     1 + x^{k(q^2+q+1)} + x^{2k(q^2+q+1)}       \right)$$
				is a triple-cycle permutation over $\gf_{q^{3}}$.
			\end{Prop}
			\begin{proof}
				First, we acquire several equations for preparations. 
				Note $7 \mid (q-1)$ by $ q =2^{3m}$.
				For any $y \in \gf_8$, we have
				\begin{equation}\label{asfwedwed}
					\begin{aligned}
						\left( 1+y+ y^{3 } + y^{5 } +y^{6 }  \right) ^{3} =& \left(  1+y+ y^{3 } + y^{5 } +y^{6 }  \right)  \left(    1+y^2+ y^{3 } + y^{5 } +y^{6 }         \right) \\
						=&  1+y^{2} +y^{4}   ,\\
					\end{aligned}
				\end{equation}
				and
				\begin{equation}\label{asvasfawsdqwdqwdasd}
					\begin{aligned}
						\left(  1+y+ y^{3}\right)^{3} =&	\left(  1+y+ y^{3}\right)  \left(  1+y^2+ y^6  \right)    \\
						=&1+y+y^2+y^3+y^4.  \\
					\end{aligned}
				\end{equation}
				Clearly, for any $x \in \gf_{q}^*$, we have $x^k \in \gf_8^*$ by $ 7k \equiv 0 \pmod{q-1}   $.
				Let $\sigma(x)= 1 + x^{k} + x^{3 k}+ x^{5 k}+ x^{6 k} $.
				According to Eq. (\ref{asfwedwed}) and Eq. (\ref{asvasfawsdqwdqwdasd}) respectively, we obtain that for any $x \in \gf_{q}^*$,
				\begin{equation}\label{asvewefwe}
					\sigma(x)^{k} = 1+x^{2k}+x^{4k} 
				\end{equation}
				and
				\begin{equation}\label{svefwfqedw}
					\left(1+x^{3 k}+x^{6 k}\right)^k=  1 +x^k +x^{2 k}+x^{3 k} + x^{4 k} ,
				\end{equation}
				since $k \equiv 3 \pmod 7$.
				Then, by raising Eq. (\ref{asvewefwe}) to the power of $2,3,5$ and $6$ respectively, we acquire
				\begin{equation}\label{s4y5u64htgsf}
					\begin{aligned}
					\sigma(x)^{2k} =& 	\left( 1+x^{2k}+ x^{4k}  \right)^{2}       	= 1+x^{ k}+x^{4 k}, \\
					\end{aligned}
				\end{equation}
				\begin{equation}\label{afdwefasdfgv34}
					\begin{aligned}
						\sigma(x)^{3k} 	=&   \left(  1 +x^{2k} + x^{4k} \right)  \left(1 +x^{2k}  + x^{4k} \right)^{2}         \\
						=&      \left( 1+x^{2k} + x^{4k} \right)    \left(  1 + x^k+  x^{4 k} \right)   \\
						=& 1+ x^{2k} +x^{3k} + x^{5k}+    x^{6k}         ,\\
					\end{aligned}
				\end{equation}
				\begin{equation}\label{fwer2t43t3e2g}
					\begin{aligned}
						\sigma(x)^{5k} =&  \left( 1+x^{2k}+  x^{4k}  \right)^{4} \left( 1+x^{2k}  + x^{4k} \right)	  \\
						=&  1 +x^{k}  +x^{3k}+x^{5k} +  x^{6k} ,  \\
					\end{aligned}
				\end{equation}
				\begin{equation}\label{fsert4whtb}
					\begin{aligned}
						\sigma(x)^{6k} =&	\left(  x^{6k}   + x^{5k}   +x^{3k}    + x^{2k}+1\right)^{2}  \\
						=&  1  +x^{3k} +x^{4k} +x^{5k}+ x^{6k}.  \\
					\end{aligned}
				\end{equation}
				And, by raising Eq. (\ref{svefwfqedw}) to the power of $2$, we get
				\begin{equation}\label{awdaefwrweqf}
					\left( 1 +x^{3 k} +  x^{6 k}\right)^{2k}=1  +x^{ k}  +x^{2 k} +x^{4 k} + x^{6 k}.
				\end{equation}
				
				After the preparation above, we now prove the theorem by Corollary \ref{mulcore}.
				Let $h(x)= 1+x^k+x^{2k}$ and
				\begin{equation}\label{asdwefeqd}
					g(x)  = x  h(x)^{q^2+q+1}   = x  (1+x^k+x^{2k}  )^{q^2+q+1}  .
				\end{equation}
				Then, $f(x)$ can be written as $	x h\left( x^{q^2+q+1}  \right)$.
				To apply Corollary \ref{mulcore}, we will compute $g\left(g(x) \right) $ and verify $g\left( g\left(g(x) \right)\right) =x $ in the below, for any $x\in \gf_{q}^*$.
				After that, we verify $\varphi(x)=h(x)h\left( g\left( x\right) \right)h\left( g\left( g(x) \right) \right)  =1 $, for any $x \in \gf_q^*$.
				
				By expanding Eq. (\ref{asdwefeqd}), one can obtain for each $x \in \gf_{q}^*$,
				\begin{equation}
					\begin{aligned}
						g(x)  = x  \left( 1+x^{k}+x^{2k}  \right)^3=&	 x \left( 1+x^{k}+x^{2k}  \right)  \left( 1+x^{2k}+x^{4k}  \right)    \\
						=& x \left(  1 +  x^k+ x^{3 k}+ x^{5 k}+x^{6 k}\right)  .\\
					   	=& x \sigma(x) .\\
					\end{aligned}
				\end{equation}
				Thus, we have
				\begin{equation}\label{awdqwfedqwd}
					\begin{aligned}
						g\left(g(x) \right)=& x  \sigma(x) \left(  1 + x^k \sigma(x)^k  +  x^{3k} \sigma(x)^{3k}+ x^{5k} \sigma(x)^{5k} +     x^{6k} \sigma(x)^{6k}   \right) .\\
					\end{aligned}
				\end{equation}
				By plugging Eqs. (\ref{s4y5u64htgsf}), (\ref{afdwefasdfgv34}), (\ref{fwer2t43t3e2g}) and (\ref{fsert4whtb}) into Eq. (\ref{awdqwfedqwd}), one can arrive at
				\begin{equation}\label{ggx}
					g(g(x))  = x  \sigma(x)  \left(  1 + x^{k}+ x^{6k}  \right) = x \left( 1 +x^{3 k}+x^{6 k}\right).
				\end{equation}
				By plugging Eq. (\ref{ggx}) into $g(x)=x \left(  1+  x^k+ x^{3 k}+ x^{5 k}+x^{6 k}\right) $, one can obtain
				\begin{equation}\label{ahdbwquyquygdiuq}
					g(g(g(x)))  = x  \sigma(x)   \left( 1+   x^{3k} \sigma(x)^{3k} + x^{6k} \sigma(x)^{6k} \right)   .
				\end{equation}
				After plugging Eqs. (\ref{afdwefasdfgv34}) and (\ref{fsert4whtb}) into Eq. (\ref{ahdbwquyquygdiuq}), we have
				\begin{equation*}\label{}
					\begin{aligned}
						g(g(g(x)))  	=&	x  (  1+ x^k+ x^{3 k}+ x^{5 k}+x^{6 k} )  \left( 1+   x^{k}  +x^{4k}  \right)  \\
						=& x (  x^{2k}+ x^{4 k}+ x^{6 k}+ 1 + x^{k}       +         x^{5k}+ 1+ x^{2 k}+x^{3k}+x^{4k}     \\
						&  +          x^k+ x^{3 k}+ x^{5 k}+x^{6 k}+1         ) \\
						=&x,  \\
					\end{aligned}
				\end{equation*}
				for each $x\in \gf_{q}^*$.
				
				Finally, for each $x\in \gf_{q}^*$, we have
				\begin{equation}\label{asdasfwqfqegbhgdhne}
					\begin{aligned}
						\varphi(x)=& (1+x^k+x^{2k}  )  h\left( x \sigma(x) \right)  h\left(  x (1+ x^{3 k}+x^{6 k})     \right) \\
						=& (1+x^k+x^{2k}  )  \left(     1+     x^{k} \sigma(x)^{k} +  x^{2k} \sigma(x)^{2k}                  \right) \\
						&  \cdot   \left(    1+(x ( 1+  x^{3 k}+x^{6 k})  )^k+(x ( 1+  x^{3 k}+x^{6 k})  )^{2k}                        \right)  .
					\end{aligned}
				\end{equation}
				By plugging Eqs.  (\ref{svefwfqedw}), (\ref{s4y5u64htgsf}),  (\ref{afdwefasdfgv34}), (\ref{fwer2t43t3e2g}), (\ref{fsert4whtb}) and (\ref{awdaefwrweqf}) into Eq. (\ref{asdasfwqfqegbhgdhne}) and simplifying it, one can obtain
				\begin{equation}\label{sadfwefgwrgth45}
					\varphi(x)=(1+x^k+x^{2k}  )  \left(     1+     x^{k}+x^{5k}    +        x^{2 k}+x^{ 6k}             \right)  \left(    1 +x^{5k}  +x^{6 k}                 \right) .
				\end{equation}
				After expanding Eq. (\ref{sadfwefgwrgth45}), one will get
				\begin{equation}
					\begin{aligned}
						\varphi(x)	=&  \left(   1+x^k+x^{2k}  \right)   \left( 1+x^{2 k}+x^{3 k}+x^{5 k}+x^{6 k}\right) 	 \\
						=&  1+x^{2 k}+x^{3 k}+x^{5 k}+x^{6 k}   +   x^k+x^{3 k}+x^{4 k}+x^{6 k}+ 1  +          x^{2k}+x^{4 k}+x^{5 k}+1+x^{8 k}  \\
						=& 1 .\\
					\end{aligned}
				\end{equation}
				Thus, $f$ is a triple-cycle permutation over $\gf_{q^{3}}$, according to Corollary \ref{mulcore}.
			\end{proof}
			\begin{Exa}
				Let $q=2^6$.
				Then $k=45$ satisfies $ 7\times 45 \equiv 0 \pmod{63}   $ and $45 \equiv 3 \pmod 7$.
				Thus, 
				$f(x)=x \left(     1 + x^{45\times4161} + x^{90 \times4161}       \right)$
				is a triple-cycle permutation over $\gf_{2^{18}}$.
			\end{Exa}
			%

			\begin{Th}
				\label{aqwefvgbvfcqx}
				Let $ q $ be a prime power.
				Assume for any $x\in \mu_{q^2+q+1}$, $ h(x) ^{q-1}=1$.
				Then $f(x)=x^q h\left(x^{q-1}\right)$ is a triple-cycle permutation over $\gf_{q^{3}}$ if and only if $ h(x) h( x^q )   h\left(    x^{q^2}  \right)  =1$ holds for any $x \in \mu_{q^2+q+1}$.
			\end{Th}
			\begin{proof}
				The proof is easy by Lemma \ref{single}, and we omit it.
			\end{proof}
			
			\begin{Prop}
				\label{}
				Let $ q=2^{2m} $.
				Assume $h(x)=1+\alpha x^{\frac{q^2+q+1}{3}}   +  x^\frac{2q^2+2q+2}{3}  $, where $\alpha \in \gf_{q}$ satisfying $\alpha^3=1$.
				Then $f(x)=x^q h\left(x^{q-1}\right)$ is a triple-cycle permutation over $\gf_{q^{3}}$.
			\end{Prop}
			\begin{proof}
				Clearly $3 \mid (q^2+q+1)$ by $ q=2^{2m} $.
				For each  $x\in \mu_{q^2+q+1}$, let $ y= x^{\frac{q^2+q+1}{3}}$.
				We have $y^3=1$ and $y^q=y$.
				
				Below, we prove that $h(x)^3 =1$ for any $x\in \mu_{q^2+q+1}$.
				After expanding and simplifying $( 1+\alpha y   +  y^2 )^3$, we obtain
				\begin{equation}\label{iugbuygvut}
					( 1+\alpha y   +  y^2 )^3= \alpha ^2 y    +\alpha ^3       + \alpha ^2 y^2+ \alpha  y^2+ \alpha  y+1 + y   + y^2+1.
				\end{equation}
				Eq. (\ref{iugbuygvut}) equals to
				\begin{equation}\label{sgafwefeqdd}
					\left(\alpha ^2+\alpha +1\right) y^2+\left(\alpha ^2+\alpha+1\right) y+\alpha ^3=1.
				\end{equation}
				
				Then, we have $h(x)^{q-1} =1$ due to $ 3 \mid (q-1)$.
				Note $\alpha \in \gf_{q}$.
				Thus $h(x^q)=h(x)^q=h(x)$, which leads to
				\begin{equation}\label{avafawwqd}
					h(x) h( x^q )   h\left(    x^{q^2}  \right) =     h(x)^3=1 .
				\end{equation}
				Thus, $f$ is a triple-cycle permutation over $\gf_{q^{3}}$, according to Theorem \ref{aqwefvgbvfcqx}.
			\end{proof}
			

			\begin{Th}
				\label{jieguo3aadafqeqgf}
				Let $ q $ be a prime power, $\phi(x) \in \mathbb{F}_{q^{2}}[x]$ and $ h(x)= \phi(x)+x^{(v-1)q}   \phi(x)^q + \psi(x)    $, where $v^3 \equiv 1 \pmod{q+1}$ and  $\psi(x)$ satisfying $\psi(x)^{q-1}=x^{v-1}$.
				Then $f(x)=x h\left(x^{q-1}\right)$ is a triple-cycle permutation over $\gf_{q^{2}}$ if and only if 	$ h(x)h( x^v ) h\left(    x^{v^2}  \right)  =1 $ holds for any $x \in \mu_{q+1}$.
			\end{Th}
			\begin{proof}
				If there exists an $ x_0 \in \mu_{q+1} $ such that $ h(x_0) = 0 $.
				Then $ h(x_0)^3  =0 \ne 1 $.
				Furthermore, $f(x_0)=x_0 h\left(x_0^{q-1}\right) =0$, thus $ f $ is not a triple-cycle permutation.
				
				If for any $ x\in \mu_{q+1} $, $ h(x) \ne 0 $.
				Then we have
				$$
				h(x)^{q-1}=\frac{ \phi(x)^q   +x^{(v-1)}   \phi(x) + \psi(x)^q    }{ \phi(x)+x^{(v-1)q}   \phi(x)^q + \psi(x)    } =    x^{v-1}.
				$$
				This lead to $g(x)=x h(x)^{q-1}=x^v$, which is a triple-cycle permutation over $\mu_{q+1}$.
				Then by plugging $a=1, r=1,v=1, s=q-1$ and $g(x)= x$ into the condition in Lemma \ref{single}, we have  $f(x)$ is a triple-cycle permutation over $\mathbb{F}_{q^{2}}$ if and only if
				$ h(x)h( x^v ) h\left(    x^{v^2}  \right)  =1 $ for any $ x \in \mu_{q+1} $.
			\end{proof}
			\begin{Prop}
				\label{aasdsaasd}
				Let $ q=2^{12k-6} $ and $ - 6 t + 12 t^2 - 8 t^3 \equiv 0 \pmod{q+1}$, where $k$ is a positive integer.
				Assume $ m $ is an integer such that
				\begin{subequations}\label{YY}
\begin{align} 
		-3 m + 6 m t - 4 m t^2  \equiv &  0   \pmod{q+1}    \label{yong1}, \\
		-m - m t +  t +  t^2  \equiv  & 0  \pmod{q+1}    \label{tiaojian2}   \text{ and} \\
		13m-13t	  \equiv &  0  \pmod{q+1}    \label{tiaojian3}
	\end{align}
\end{subequations}
are all established.			Then $f(x)=x h\left(x^{q-1}\right)$ is a triple-cycle permutation over $\gf_{q^2}$, where $$ h(x)= x^m+x^{mq-2tq}  + x^t    .$$
			\end{Prop}
			\begin{proof}
				In this proof, first we will derive some useful congruences from Congruences (\ref{yong1}), (\ref{tiaojian2}), (\ref{tiaojian3}).
				Then, we will further handle these useful congruences (Congruences (\ref{yong1}), (\ref{yong2}), (\ref{yong3}) and (\ref{yong4})) to obtain Congruences (\ref{ap1}), (\ref{ap2}) , (\ref{ap3})  and (\ref{ap4}) that can be directly used to prove $f$ being triple-cycle.
				
				We have $13 \mid (q+1)$ by $  q=2^{12k-6} $.
				By simplifying $ 2 \times $ Congruence (\ref{tiaojian2}) $ - $ Congruence (\ref{yong1}), we have
				\begin{equation}
					m + 2 t - 8 m t + 2 t^2 + 4 m t^2   \equiv  0   \pmod{q+1}  \label{yong2}.
				\end{equation}
				Then, by simplifying $ 8 \times $ Congruence (\ref{tiaojian2}) $ + $ Congruence (\ref{tiaojian3}), we have
				\begin{equation}
					5 t-5 m+8 m t-8 t^{2}   \equiv  0  \pmod{q+1}  .  \label{linshi}
				\end{equation}
				After simplifying $ (-2) \times $ Congruence (\ref{yong1}) $ + $ Congruence (\ref{linshi}), one can obtain
				\begin{equation}
					m + 5 t - 4 m t - 8 t^2 + 8 m t^2 \equiv   0  \pmod{q+1}\label{yong3}   .
				\end{equation}
				By simplifying $ - $ Congruence (\ref{yong3}) $ - $ Congruence (\ref{tiaojian2}), we get
				\begin{equation}
					m - 7 t + 6 m t + 6 t^2 - 8 m t^2	  \equiv   0  \pmod{q+1} \label{yong4} .
				\end{equation}
				
				After obtaining Congruences (\ref{yong1}), (\ref{yong2}), (\ref{yong3}) and (\ref{yong4}), we will further handle them below.
				Assume $v=1-2t$.
				Then, $v^2+v+1=3 - 6 t + 4 t^2$ and according to Congruence (\ref{yong1}), we arrive at

				\begin{equation} \label{ap1}
					\left\{
					\begin{aligned}
					&	m(v^2+v+1) \equiv v(v^2+v+1) \equiv   -3 m + 6 m t - 4 m t^2 \equiv    0   \pmod{q+1}   , \\    
					&	-m v^2 - m + t v + v^2 - v  \equiv   m v + t v^2 + t   \pmod{q+1}    , \\
					&	-m v - m + t v^2 - v^2 + 1	  \equiv   m v^2 + t v + t        \pmod{q+1} , \\
					&	m v^2+m v+t     \equiv     -m+t v^2+t v-v+1      \pmod{q+1} , \\
					&	-m v^2 - m v + t + v - 1   \equiv     m + t v^2 + t v        \pmod{q+1} , \\
					&	m v + m + t v^2   \equiv  -m v^2 + t v + t + v^2 - 1        \pmod{q+1} \text{ and} \\ 
					&	-m v + t v^2 + t - v^2 + v  		  \equiv     m v^2 + m + t v        \pmod{q+1} . \\
					\end{aligned}
					\right.
				\end{equation}
				According to Congruence (\ref{yong2}), one can get
				\begin{equation}\label{ap2}
					-m v + m + t v^2 - v^2 + v \equiv m v^2 + m v - m - v + 1.
				\end{equation}
				According to Congruence (\ref{yong3}), we obtain
				\begin{equation}\label{ap3}
					m v^{2}-m v+m-v^{2}+v\equiv -m v^{2}+m v+t+v^{2}-1.
				\end{equation}
				According to Congruence (\ref{yong4}), we have
				\begin{equation}\label{ap4}
					-m v^2+m v+m+v^2-1   \equiv   m v^2-m+t v-v+1  .
				\end{equation}

				Finally, we expand $	h(x)h( x^v ) h\left(    x^{v^2}  \right) $ for any $ x \in \mu_{q+1} $ and get
				\begin{equation*}
					\begin{aligned}
						&  \left( x^m+x^{m q+(v-1)q}  + x^t \right)  \left( x^{m v}+x^{(m q+(v-1)q)v}  + x^{tv} \right) \left( x^{m v^2}+x^{(m q+(v-1)q)v^2}  + x^{tv^2} \right)    \\
						=&x^{-m v^2-m v+t+v-1}+x^{-m v^2+m v+t+v^2-1}+x^{-m v^2+m+t v+v^2-1}+x^{-m v^2+t v+t+v^2-1}\\
						&+x^{-m v^2-m+t v+v^2-v}+x^{m v^2+m v+t}+x^{m v^2+m+t v}+x^{m v^2+t v+t}+x^{m	v^2-m+t v-v+1}\\
						&+x^{m v^2-m v+t-v^2+v}+x^{m v+m+t v^2}+x^{m v+t v^2+t} +x^{m v-m+t v^2-v+1}+x^{m+t v^2+t v}\\
						&+x^{-m+t v^2+t v-v+1}+x^{-m v-m+t v^2-v^2+1}+x^{-m v+m+tv^2-v^2+v}+x^{-m v+t v^2+t-v^2+v} \\
						&+x^{-m v^2-m v-m}+x^{-m v^2-m v+m+v-1}+x^{-m v^2+m v+m+v^2-1}+x^{-m v^2+m v-m+v^2-v} \\
						&+x^{m v^2+m v+m}+x^{m v^2+m v-m-v+1}+x^{m v^2-av-m-v^2+1}+x^{m v^2-m v+m-v^2+v}+x^{t v^2+t v+t} \\
						=& 1, \\
					\end{aligned}
				\end{equation*}
				where the last equation holds by Congruences (\ref{ap1}), (\ref{ap2}) , (\ref{ap3})  and (\ref{ap4}).
				Thus $ f $ is a triple-cycle permutation. 
			\end{proof}
			\begin{Exa}
				Let $q=2^6, t=25,m=5$ in Proposition \ref{aasdsaasd}.
				Then Congruences (\ref{yong1}), (\ref{tiaojian2}) and (\ref{tiaojian3}) are all satisfied.
				Then, according to Proposition \ref{aasdsaasd}, $f(x)=x h\left(x^{q-1}\right)$ is a triple-cycle permutation over $\gf_{2^{12}}$, where $ h(x)= x^5+x^{45}  + x^{25}    .$
				This is also verified by Magma.
			\end{Exa}

			
			In the end of this paper, we provide an explicit triple-cycle permutation of the form $x^t + c\tr_{q^m/q}(x^s)$.
			%
			\begin{Prop}
				Let $q=2^{2 k}$, where $k$ is a positive integer. Put $\theta \in \mathbb{F}_{q}$ satisfying $\theta^{3}=1, \theta \neq 1$.
				Then, the compositional inverse of
				$$
				f(x)=x+\theta \operatorname{Tr}_{q^{3} / q}\left(x^{\frac{q^2+q}{2}}\right)
				$$
				is 	$		f^{-1}(x)=x +        \theta^2 \tr_{q^3/q}(x^\frac{q^2+q}{2}).		$
				Furthermore, $f$ is a triple-cycle permutation over $\mathbb{F}_{q^{3}}$.
			\end{Prop}
			\begin{proof}
				Clearly $ \theta+\theta^2=1$ and $ \theta^{q/2}=\theta^2$.
				Then, after expanding $\left(   x + \theta \tr_{q^3/q}(x^\frac{q^2+q}{2})  \right) ^\frac{q^2+q}{2}$ and simplifying it, we have
				\begin{equation}\label{zz1}
					\begin{aligned}
						\left(   x + \theta \tr_{q^3/q}(x^\frac{q^2+q}{2}) \right) ^\frac{q^2+q}{2} =&  \left(   x^{q^2/2 } +  \theta^{q^2/2} \tr_{q^3/q}(x^\frac{q^2+q}{2})^{q^2/2}  \right)        \left(   x^{q/2 } +  \theta^{q/2} \tr_{q^3/q}(x^\frac{q^2+q}{2})^{q/2}  \right)       \\
						=& x^{\frac{q^2+q}{2} }  +      x^{q^2/2 }     \theta^{2} \tr_{q^3/q}(x^\frac{q^2+q}{2})^{q/2} +   x^{q/2 }   \theta^{2} \tr_{q^3/q}(x^\frac{q^2+q}{2})^{q/2}      \\
						&+  \theta \tr_{q^3/q}(x^\frac{q^2+q}{2}). \\
					\end{aligned}
				\end{equation}
It is clear that
			\begin{equation*}\label{zero}
				\tr_{q^3/q}(x^{q^2/2} + x^{q/2 }    )    =0  .
			\end{equation*}
Thus, 
		\begin{equation}\label{zz1kuo}
		\begin{aligned}			
			 \tr_{q^3/q} \left(     \left(   x + \theta \tr_{q^3/q}(x^\frac{q^2+q}{2})  \right)^\frac{q^2+q}{2}      \right)  =   & 	 \tr_{q^3/q} \left(    x^{\frac{q^2+q}{2} }       +  \theta \tr_{q^3/q}(x^\frac{q^2+q}{2})     \right) \\
			 =   & 	 \tr_{q^3/q} (    x^{\frac{q^2+q}{2} }   )     +        \theta \tr_{q^3/q}(x^\frac{q^2+q}{2}) . \\
		\end{aligned}		
		\end{equation}
	
				We consider $f(f(x))$ for any $x\in \gf_{q^{3}}$.
				After plugging Eq. (\ref{zz1kuo}) into 
				$$f(f(x))=x + \theta \tr_{q^3/q}(x^\frac{q^2+q}{2})  + \theta \tr_{q^3/q}\left( \left( x + \theta \tr_{q^3/q}(x^\frac{q^2+q}{2})\right)^\frac{q^2+q}{2}\right) , $$
				we have 	
					\begin{equation}\label{iyvbsaaxiuhyb}
f(f(x))=x +        \theta^2 \tr_{q^3/q}(x^\frac{q^2+q}{2}).
						\end{equation}
					
				Then, we simplify $f(f(f(x)))$ for any $x\in \gf_{q^{3}}$ below.
				By plugging Eq. (\ref{zz1kuo}) into
\begin{equation*}\label{ferghfgf}		
f(f(f(x)))=  x+\theta \operatorname{Tr}_{q^{3} / q}(x^{\frac{q^2+q}{2}})      +        \theta^2 \tr_{q^3/q}\left( \left( x+\theta \operatorname{Tr}_{q^{3} / q}(x^{\frac{q^2+q}{2}})  \right)^\frac{q^2+q}{2}   \right) ,
\end{equation*}		
one can obtain 
\begin{equation*}\label{fergasdqhfgf}		
			\begin{aligned}		
	f(f(f(x)))= & x+\theta \tr_{q^{3} / q}(x^{\frac{q^2+q}{2}})      +        \theta^2  \left(   \tr_{q^3/q} (    x^{\frac{q^2+q}{2} }   )     +        \theta \tr_{q^3/q}(x^\frac{q^2+q}{2}) \right)     \\
	=&x+\theta \tr_{q^{3} / q}(x^{\frac{q^2+q}{2}})      +        \theta^2   \tr_{q^3/q} (    x^{\frac{q^2+q}{2} }   )     +         \tr_{q^3/q}(x^\frac{q^2+q}{2})   \\
			=&x. \\
	\end{aligned}		
\end{equation*}		
Thus, $		f^{-1}(x)=x +   \theta^2 \tr_{q^3/q}(x^\frac{q^2+q}{2}) 	$, and $f$ is a triple-cycle permutation over $\mathbb{F}_{q^{3}}$.
\end{proof}


\begin{thebibliography}{45}
	\providecommand{\natexlab}[1]{#1}
	\providecommand{\url}[1]{\texttt{#1}}
	\expandafter\ifx\csname urlstyle\endcsname\relax
	\providecommand{\doi}[1]{doi: #1}\else
	\providecommand{\doi}{doi: \begingroup \urlstyle{rm}\Url}\fi
	
	\bibitem[Akbary et~al.(2011)Akbary, Ghioca, and Wang]{akbary2011constructing}
	Amir Akbary, Dragos Ghioca, and Qiang Wang.
	\newblock On constructing permutations of finite fields.
	\newblock \emph{Finite Fields and Their Applications}, 17\penalty0
	(1):\penalty0 51--67, 2011.
	
	\bibitem[Anbar et~al.(2018)Anbar, Od{\v{z}}ak, Patel, Quoos, Somoza, and
	Topuzo{\u{g}}lu]{anbar2018carlitz}
	Nurdag{\"u}l Anbar, Almasa Od{\v{z}}ak, Vandita Patel, Luciane Quoos, Anna
	Somoza, and Alev Topuzo{\u{g}}lu.
	\newblock On the Carlitz rank of permutation polynomials over finite fields:
	recent developments.
	\newblock \emph{Women in numbers Europe II}, pages 39--55, 2018.
	
	\bibitem[Barreto(2000)]{barreto2000anubis}
	Paulo Barreto.
	\newblock The Anubis block cipher.
	\newblock \emph{NESSIE}, 2000.
	
	\bibitem[Barreto and Rijmen(2000)]{barreto2000khazad}
	Paulo Barreto and Vincent Rijmen.
	\newblock The Khazad legacy-level block cipher.
	\newblock \emph{Primitive submitted to NESSIE}, 97, 2000.
	
	\bibitem[Bartoli(2021)]{bartoli2021hasse}
	Daniele Bartoli.
	\newblock Hasse-Weil type theorems and relevant classes of polynomial
	functions.
	\newblock \emph{Surveys in Combinatorics 2021}, 470:\penalty0 43, 2021.
	
	\bibitem[Bartoli and Timpanella(2021)]{bartoliFamilyPermutationTrinomials2021a}
	Daniele Bartoli and Marco Timpanella.
	\newblock A family of permutation trinomials over $\mathbb{F}_{q^2}$.
	\newblock \emph{Finite Fields and Their Applications}, 70:\penalty0 101781,
	2021.
	
	\bibitem[Borghoff et~al.(2012)Borghoff, Canteaut, G{\"u}neysu, Kavun, Knezevic,
	Knudsen, Leander, Nikov, Paar, Rechberger, et~al.]{borghoff2012prince}
	Julia Borghoff, Anne Canteaut, Tim G{\"u}neysu, Elif~Bilge Kavun, Miroslav
	Knezevic, Lars~R Knudsen, Gregor Leander, Ventzislav Nikov, Christof Paar,
	Christian Rechberger, et~al.
	\newblock Prince--a low-latency block cipher for pervasive computing
	applications.
	\newblock In \emph{International Conference on the Theory and Application of
		Cryptology and Information Security}, pages 208--225. Springer, 2012.
	
	\bibitem[Canteaut and Rou{\'e}(2015)]{canteaut2015behaviors}
	Anne Canteaut and Jo{\"e}lle Rou{\'e}.
	\newblock On the behaviors of affine equivalent {S}-boxes regarding
	differential and linear attacks.
	\newblock In \emph{Annual International Conference on the Theory and
		Applications of Cryptographic Techniques}, pages 45--74. Springer, 2015.
	
	\bibitem[Charpin and Kyureghyan(2010)]{charpin2010monomial}
	Pascale Charpin and Gohar Kyureghyan.
	\newblock Monomial functions with linear structure and permutation polynomials.
	\newblock In \emph{Finite fields: theory and applications}, volume 518, pages
	99--111. AMS Providence, RI, USA, 2010.
	
	\bibitem[Charpin et~al.(2016)Charpin, Mesnager, and
	Sarkar]{charpin2016involutions}
	Pascale Charpin, Sihem Mesnager, and Sumanta Sarkar.
	\newblock Involutions over the Galois Field $ \gf_{2^n} $.
	\newblock \emph{IEEE Transactions on Information Theory}, 62\penalty0
	(4):\penalty0 2266--2276, 2016.
	
	\bibitem[Chen et~al.(2021)Chen, Wang, and
	Zhu]{chenConstructionsNcyclePermutations2021}
	Yuting Chen, Liqi Wang, and Shixin Zhu.
	\newblock On the constructions of $n$-cycle permutations.
	\newblock \emph{Finite Fields and Their Applications}, 73:\penalty0 101847,
	2021.
	
	\bibitem[Coulter and Mesnager(2018)]{coulter2018bent}
	Robert~S. Coulter and Sihem Mesnager.
	\newblock Bent functions from involutions over $\gf_{2^{n}}$.
	\newblock \emph{IEEE Transactions on Information Theory}, 64\penalty0
	(4):\penalty0 2979--2986, 2018.
	
	\bibitem[Daemen and Rijmen(2013)]{daemen2013design}
	Joan Daemen and Vincent Rijmen.
	\newblock \emph{The design of Rijndael: AES-the advanced encryption standard}.
	\newblock Springer Science \& Business Media, 2013.
	
	\bibitem[Ding et~al.(2015)Ding, Qu, Wang, Yuan, and Yuan]{ding2015permutation}
	Cunsheng Ding, Longjiang Qu, Qiang Wang, Jin Yuan, and Pingzhi Yuan.
	\newblock Permutation trinomials over finite fields with even characteristic.
	\newblock \emph{SIAM Journal on Discrete Mathematics}, 29\penalty0
	(1):\penalty0 79--92, 2015.
	
	\bibitem[Gallager(1962)]{gallager1962low}
	Robert Gallager.
	\newblock Low-density parity-check codes.
	\newblock \emph{IRE Transactions on Information Theory}, 8\penalty0
	(1):\penalty0 21--28, 1962.
	
	\bibitem[Gerike and Kyureghyan(2020)]{gerikePermutationsFiniteFields2020}
	Daniel Gerike and Gohar~M. Kyureghyan.
	\newblock Permutations on finite fields with invariant cycle structure on
	lines.
	\newblock \emph{Designs, Codes and Cryptography}, 88\penalty0 (9):\penalty0
	1723--1740, 2020.
	
	\bibitem[Gupta and Sharma(2016)]{gupta2016some}
	Rohit Gupta and RK~Sharma.
	\newblock Some new classes of permutation trinomials over finite fields with
	even characteristic.
	\newblock \emph{Finite Fields and Their Applications}, 41:\penalty0 89--96,
	2016.
	
	\bibitem[Hou(2015)]{hou2015permutation}
	Xiangdong Hou.
	\newblock Permutation polynomials over finite fields—a survey of recent
	advances.
	\newblock \emph{Finite Fields and Their Applications}, 32:\penalty0 82--119,
	2015.
	
	\bibitem[Hou et~al.(2020)Hou, Tu, and
	Zeng]{houDeterminationClassPermutation2020}
	Xiangdong Hou, Ziran Tu, and Xiangyong Zeng.
	\newblock Determination of a class of permutation trinomials in characteristic
	three.
	\newblock \emph{Finite Fields and Their Applications}, 61:\penalty0 101596,
	2020.
	
	\bibitem[Kyureghyan and Zieve(2016)]{kyureghyan2016permutation}
	Gohar Kyureghyan and Michael Zieve.
	\newblock Permutation polynomials of the form $x+ \gamma \tr(x^k)$.
	\newblock In \emph{Contemporary developments in finite fields and
		applications}, pages 178--194. World Scientific Singapore, 2016.
	
	\bibitem[Li et~al.(2017)Li, Qu, and Chen]{li2017new}
	Kangquan Li, Longjiang Qu, and Xi~Chen.
	\newblock New classes of permutation binomials and permutation trinomials over
	finite fields.
	\newblock \emph{Finite Fields and Their Applications}, 43:\penalty0 69--85,
	2017.
	
	\bibitem[Li et~al.(2018{\natexlab{a}})Li, Qu, Chen, and Li]{li2018permutation}
	Kangquan Li, Longjiang Qu, Xi~Chen, and Chao Li.
	\newblock Permutation polynomials of the form $ cx + \tr_{q^n/q}(x^a) $ and
	permutation trinomials over finite fields with even characteristic.
	\newblock \emph{Cryptography and Communications}, 10\penalty0 (3):\penalty0
	531--554, 2018{\natexlab{a}}.
	
	\bibitem[Li et~al.(2018{\natexlab{b}})Li, Qu, Li, and Fu]{li2018newp}
	Kangquan Li, Longjiang Qu, Chao Li, and Shaojing Fu.
	\newblock New permutation trinomials constructed from fractional polynomials.
	\newblock \emph{Acta Arithmetica}, 183:\penalty0 101--116, 2018{\natexlab{b}}.
	
	\bibitem[Li et~al.(2021)Li, Wang, Xu, and Zeng]{liSeveralClassesComplete2021}
	Lisha Li, Qiang Wang, Yunge Xu, and Xiangyong Zeng.
	\newblock Several classes of complete permutation polynomials with {Niho}
	exponents.
	\newblock \emph{Finite Fields and Their Applications}, 72:\penalty0 101831,
	2021.
	
	\bibitem[Li(2017)]{li2017two}
	Nian Li.
	\newblock On two conjectures about permutation trinomials over $\gf_{3^{2k}}$.
	\newblock \emph{Finite Fields and Their Applications}, 47:\penalty0 1--10,
	2017.
	
	\bibitem[Li and Helleseth(2017)]{li2017several}
	Nian Li and Tor Helleseth.
	\newblock Several classes of permutation trinomials from {Niho} exponents.
	\newblock \emph{Cryptography and Communications}, 9\penalty0 (6):\penalty0
	693--705, 2017.
	
	\bibitem[Li and Helleseth(2019)]{liNewPermutationTrinomials2019b}
	Nian Li and Tor Helleseth.
	\newblock New permutation trinomials from {Niho} exponents over finite fields
	with even characteristic.
	\newblock \emph{Cryptography and Communications}, 11\penalty0 (1):\penalty0
	129--136, 2019.
	
	\bibitem[Li and Zeng(2019)]{li2018survey}
	Nian Li and Xiangyong Zeng.
	\newblock A survey on the applications of {Niho} exponents.
	\newblock \emph{Cryptography and Communications}, 11\penalty0 (3):\penalty0
	509--548, 2019.
	
	\bibitem[Liu et~al.(2019)Liu, Chen, Xu, and
	Sun]{liuTripleCyclePermutationsFinite2019}
	Xianping Liu, Yuan Chen, Yunge Xu, and Zhimin Sun.
	\newblock Triple-{cycle} permutations over finite fields of characteristic two.
	\newblock \emph{International Journal of Foundations of Computer Science},
	30\penalty0 (2):\penalty0 275--292, 2019.
	
	\bibitem[Mullen and Wang(2014)]{MullenWang14}
	Gary~Lee Mullen and Qiang Wang.
	\newblock Permutation polynomials of one variable.
	\newblock In \emph{Handbook of Finite Fields}, pages 215--230. CRC, 2014.
	
	\bibitem[Niu et~al.(2020)Niu, Li, Qu, and Wang]{niu2019new}
	Tailin Niu, Kangquan Li, Longjiang Qu, and Qiang Wang.
	\newblock New constructions of involutions over finite fields.
	\newblock \emph{Cryptography and Communications}, 12\penalty0 (2):\penalty0
	165--185, 2020.
	
	\bibitem[Niu et~al.(2021)Niu, Li, Qu, and Wang]{2022niufinding}
	Tailin Niu, Kangquan Li, Longjiang Qu, and Qiang Wang.
	\newblock Finding compositional inverses of permutations from the {AGW}
	criterion.
	\newblock \emph{IEEE Transactions on Information Theory}, 67\penalty0
	(8):\penalty0 4975--4985, 2021.
	
	\bibitem[Tu and Zeng(2018)]{tu2018two}
	Ziran Tu and Xiangyong Zeng.
	\newblock Two classes of permutation trinomials with {N}iho exponents.
	\newblock \emph{Finite Fields and Their Applications}, 53:\penalty0 99--112,
	2018.
	
	\bibitem[Tuxanidy and Wang(2014)]{tuxanidy2014inverses}
	Aleksandr Tuxanidy and Qiang Wang.
	\newblock On the inverses of some classes of permutations of finite fields.
	\newblock \emph{Finite Fields and Their Applications}, 28:\penalty0 244--281,
	2014.
	
	\bibitem[Wang(2019)]{Wang2019index}
	Qiang Wang.
	\newblock Polynomials over finite fields: an index approach.
	\newblock In \emph{Combinatorics and Finite Fields. Difference Sets,
		Polynomials, Pseudorandomness and Applications}, pages 319--348. Degruyter,
	2019.
	
	\bibitem[Wu et~al.(2017)Wu, Yuan, Ding, and Ma]{wu2017permutation}
	Danyao Wu, Pingzhi Yuan, Cunsheng Ding, and Yuzhen Ma.
	\newblock Permutation trinomials over $\mathbb{F}_{2^m}$.
	\newblock \emph{Finite Fields and Their Applications}, 46:\penalty0 38--56,
	2017.
	
	\bibitem[Wu et~al.(2020)Wu, Li, and
	Wang]{wuCharacterizationsConstructionsTriplecycle2020a}
	Mengna Wu, Chengju Li, and Zilong Wang.
	\newblock Characterizations and constructions of triple-cycle permutations of
	the form $x^rh(x^s)$.
	\newblock \emph{Designs, Codes and Cryptography}, 88\penalty0 (10):\penalty0
	2119--2132, 2020.
	
	\bibitem[Xu et~al.(2018)Xu, Cao, and Ping]{xu2018some}
	Guangkui Xu, Xiwang Cao, and Jingshui Ping.
	\newblock Some permutation pentanomials over finite fields with even
	characteristic.
	\newblock \emph{Finite Fields and Their Applications}, 49:\penalty0 212--226,
	2018.
	
	\bibitem[Yuan and Ding(2007)]{yuan2007four}
	Jin Yuan and Cunsheng Ding.
	\newblock Four classes of permutation polynomials of $\mathbb{F}_{2^m}$.
	\newblock \emph{Finite fields and their applications}, 13\penalty0
	(4):\penalty0 869--876, 2007.
	
	\bibitem[Zha and Hu(2016)]{zha2016some}
	Zhengbang Zha and Lei Hu.
	\newblock Some classes of permutation polynomials of the form
	$\left(x^{p^{m}}-x+\delta\right)^{s}+x$ over $\mathbb{F}_p^{2m}$.
	\newblock \emph{Finite Fields and Their Applications}, 40:\penalty0 150--162,
	2016.
	
	\bibitem[Zha et~al.(2017)Zha, Hu, and Fan]{zha2017further}
	Zhengbang Zha, Lei Hu, and Shuqin Fan.
	\newblock Further results on permutation trinomials over finite fields with
	even characteristic.
	\newblock \emph{Finite Fields and Their Applications}, 45:\penalty0 43--52,
	2017.
	
	\bibitem[Zha et~al.(2019)Zha, Hu, and Zhang]{ZHA2019101573}
	Zhengbang Zha, Lei Hu, and Zhizheng Zhang.
	\newblock Permutation polynomials of the form $x+\gamma
	\operatorname{Tr}_{q}^{q^{n}}(h(x))$.
	\newblock \emph{Finite Fields and Their Applications}, 60:\penalty0 101573,
	2019.
	
	\bibitem[Zheng et~al.(2019{\natexlab{a}})Zheng, Yuan, Li, Hu, and
	Zeng]{zheng2019constructions}
	Dabin Zheng, Mu~Yuan, Nian Li, Lei Hu, and Xiangyong Zeng.
	\newblock Constructions of involutions over finite fields.
	\newblock \emph{IEEE Transactions on Information Theory}, 65\penalty0
	(12):\penalty0 7876--7883, 2019{\natexlab{a}}.
	
	\bibitem[Zheng et~al.(2019{\natexlab{b}})Zheng, Yuan, and Yu]{zheng2019two}
	Dabin Zheng, Mu~Yuan, and Long Yu.
	\newblock Two types of permutation polynomials with special forms.
	\newblock \emph{Finite Fields and Their Applications}, 56:\penalty0 1--16,
	2019{\natexlab{b}}.
	
	\bibitem[Zheng et~al.(2021)Zheng, Kan, Peng, and
	Tang]{zhengTwoClassesPermutation2021a}
	Lijing Zheng, Haibin Kan, Jie Peng, and Deng Tang.
	\newblock Two classes of permutation trinomials with {Niho} exponents.
	\newblock \emph{Finite Fields and Their Applications}, 70:\penalty0 101790,
	2021.
	
\end{thebibliography}

		\end{document}